\documentclass[12pt]{article}
\usepackage{amsmath}
\usepackage{amssymb}
\usepackage{amsthm}
\usepackage{latexsym}
\usepackage{times}
\usepackage{braket}
\usepackage{mathtools}
\usepackage{dsfont}
\usepackage{graphicx}
\usepackage{subfig}
\usepackage{hyperref} 
\hypersetup{
	colorlinks=true,
	linkcolor=blue,
	urlcolor=blue,
	citecolor=blue,
}
\usepackage[capitalise]{cleveref}
\usepackage[numbers,sort&compress]{natbib}

\usepackage{etoolbox} 

\usepackage{authblk}
\usepackage{ulem}
\usepackage[thicklines]{cancel}

\newtheorem{theorem}{Theorem}

\newtheorem{corollary}{Corollary}
\newtheorem{lemma}{Lemma}

\newcommand{\abs}[1]{\left|{#1}\right|}

\title{\Large Probabilistic Coherence Transformation under Strictly Incoherent Operation}
\author[1]{Ao-Xiang Liu}
\author[1,2]{Cong-Feng Qiao\thanks{Corresponding author: \href{mailto:qiaocf@ucas.ac.cn}{qiaocf@ucas.ac.cn}}}
\affil[1]{School of Physical Sciences, University of Chinese Academy of Sciences, Beijing 100049, China}
\affil[2]{Key Laboratory of Vacuum Physics, University of Chinese Academy of Sciences, Beijing 100049, China}
\date{\normalsize Dated: \today}

\begin{document}

\baselineskip24pt
\maketitle

\begin{abstract}
\normalsize
We investigate the probabilistic coherence transformation under strictly incoherent operations on majorization lattice. To this end, the greedy and thrifty protocols are adapted for the probabilistic coherence transformation, of which the latter exhibits certain superiority in preserving coherence on average. We provide new insight into optimal common resource state by showing its natural interpretations in terms of distance measure, and formalize the intuition that the large coherence gain can be realized with the price of success probability loss for coherence transformation, and vice versa. Collectively, these results justify the optimality of the two protocols. Additionally, deterministic and probabilistic coherence transformations between two mixed states are explored. As an application, it is shown that the conversion from coherence into entanglement may benefit from probabilistic coherence transformation.
\end{abstract}

\section{Introduction}
\label{sec:intro}
Quantum coherence is one of the fundamental features of quantum theory and a valuable resource in quantum information processing \cite{MA10Q}, such as quantum computation \cite{SP97P,GL97Q}, quantum metrology \cite{GV04Q,GV11A,PD18C,XB18S,ZC19D,LR24R}, quantum channel discrimination \cite{GD13C,FA14D,SA17C,TA19O,MW20T,CS22Q}, quantum correlations \cite{MJ16C,HX16E,HX16Q,MD17N,GD17W,DZ19E,LJ21C,LK23C,MC23A}, quantum biology and transport phenomena \cite{LN13Q}, quantum phase transitions \cite{KG14Q,CB15F,MA16Q,CJ16C,LY16Q,SZ24E}, quantum cryptography \cite{BC14Q}, and quantum chaos \cite{AN21Q}. Recent years, the resource theory of quantum coherence has seen significant development. It not only provides a rigorous framework for quantifying coherence, but also offers a new perspective for understanding quantum coherence. 

For any resource theory of coherence, one of the main problems is how to describe the transformation between states via free operations \cite{EC19Q}. There are different types of definitions of free operations by specified physical considerations \cite{SA17C}, such as maximally incoherent operations \cite{AJ06Q}, incoherent operations (IOs) \cite{BT14Q,YY15Q},  strictly incoherent operations (SIOs) \cite{WA16O,YB16Q}, dephasing-covariant incoherent operations \cite{CE16C,CE16Critical,MI16H}, and genuinely incoherent operations \cite{DV17G}. The quantum states that are diagonal in a prefixed reference basis are refered as incoherent states \cite{EC19Q}. The notion of SIOs was introduced by $\mathring{\text{A}}$berg \cite{AJ06Q}, which has no business with coherence but a physical interpretation in terms of interferometry \cite{YB16Q}. The coherence transformation (CT) asks for free operation transforming one quantum state into another. This issue was solved in one-qubit case and for arbitrary pure state to pure state \cite{DS15C,DS17E,CE16C,DS15Co,ZH17O} (see \cref{prop:cohe_tran}) in terms of the majorization relation \cite{MA11I}. The necessary and sufficient conditions for transforming mixed state into pure state \cite{LC19D} and pure state to mixed state \cite{DS19C} via SIOs were accomplished. Nevertheless, the CT between two arbitrary mixed states remains unsolved generally, though some efforts have been done \cite{DS19C,BG21G,LC22A,BZ24S}. To circumvent this difficulty, one may achieve the CT with the aid of catalyst state \cite{LB24C}. Additional progress in quantum coherence see Refs. \cite{ZQ16D,RC19B,QW22C,SY23C,PJ24V,ZH24N} and the references therein. Among various coherence transformation protocols, the probabilistic coherence transformation (PCT) between states is also meaningful in coherence resource theory. In \cite{DS15Co}, an optimal coherence transformation strategy between two pure coherent states is presented following the approach that has been established for entanglement in \cite{VG99E}. Recently, the optimal probabilistic coherence distillation protocols are proposed in \cite{TG19O,LC21O,LC20C} and a framework of approximate coherence transformation are investigated in \cite{LC22A}. Usually, people employ the deterministic CT from initial state to an intermediate state, which may lead to large waste of coherence when PCT falis. 

In this work, we investigate the probabilistic coherence transformation between states under SIO.
In the framework of majorization lattice \cite{FC02S}, the greedy and thrifty protocols \cite{DS24P} were employed and proved to be optimal  for PCT, of which the thrifty protocol preserves more coherence on average. The paper is organized as follows. In \cref{sec:preliminaries}, the basic notions are introduced and the intuitive trade-off relation between the coherence gain and the maximal CT probability is proved. In \cref{sec:protocols}, the greedy and thrifty protocols of PCT for pure states are constructed, and they are compared based on majorization lattice and fidelity. In \cref{sec:protranmix}, we generalize the protocols to the case of mixed state. The probabilistic conversion from coherence into entanglement through probabilistic coherence transformation is studied in \cref{sec:coheent}. The conclusions are given in \cref{sec:con}.
\section{Preliminaries}
\label{sec:preliminaries}
In a prefixed basis of the $d$-dimensional Hilbert space, a state is incoherent (free) state if it is diagonal in the basis, otherwise it is coherent state. The set of incoherent states will be denoted by $\mathcal{I}$. A strictly incoherent operation is a completely positive trace preserving map  \cite{BT14Q,YY15Q,YB16Q,WA16O}, that is 
\begin{align}
\Lambda(\rho)=\sum_{\mu=1}^{u}K_\mu\rho K_\mu^\dagger\ .
\end{align}
Here, the Kraus operator $K_\mu$, the so called strictly incoherent operator, obeys the completeness condition, $\sum_{\mu=1}^{u}K_\mu^\dagger K_\mu=\mathbb{I}$ with $\mathbb{I}$ the identity operator on the Hilbert space, and satisfying $\sum_{\mu=1}^{u}K_\mu\mathcal{I} K_\mu^\dagger\subset\mathcal{I}$, \ $\sum_{\mu=1}^{u}K_\mu^\dagger\mathcal{I} K_\mu\subset\mathcal{I}$. It can be observed that, in each column or row of $K_\mu$, there exists at most one nonvanishing element \cite{YY15Q,DS15C}. One can show that the projector $P_\mu=\sum_i \ket{i}\bra{i}$ is incoherent, which can be denoted as a strictly incoherent projector \cite{LC22A}.

Majorization theory \cite{MA11I} is found a useful tool in the study of CT under incoherent operations, and is widely involved in the quantum information theory \cite{MA99C,MA01C,MA01M,PM11M,FS13U,PZ13M,BG17A,GG18C,LJ19T,BG19O,BG19M,LJ20A,LJ21C,YB16Q,ZG23E,XY23Q,LA23A}. 
Given two distributions $\boldsymbol{a}$ and $\boldsymbol{b}\in\mathcal{P}^{d}$, the $d$-dimensional probability distribution, $\boldsymbol{a}$ is said to be majorized by $\boldsymbol{b}$, in symbols $\boldsymbol{a}\prec\boldsymbol{b}$, if and only if $\sum_{i=1}^{k}a^\downarrow_{i}\leqslant\sum_{j=1}^{k}b^\downarrow_{j},\ k\in\lbrace1,...,d\rbrace$, where $\downarrow$ indicates that $\boldsymbol{a}$ and $\boldsymbol{b}$ are descending ordered. If neither $\boldsymbol{a}\nprec\boldsymbol{b}$ nor $\boldsymbol{b}\nprec\boldsymbol{a}$, then $\boldsymbol{a}$ and $\boldsymbol{b}$ are incomparable. The majorization relation equipped with join ($\vee$) and meet ($\wedge$) operations constitutes a majorization lattice \cite{FC02S}, in which $\boldsymbol{a}\vee\boldsymbol{b}\in\mathcal{P}^d$ and $\boldsymbol{a}\wedge\boldsymbol{b}\in\mathcal{P}^d$ are the unique least upper bound and greatest lower bound of $\boldsymbol{a}$, $\boldsymbol{b}$ up to a permutation transformation. Similarly, we write the least upper and greatest lower bounds of set $\mathcal{S}\subset\mathcal{P}^d$ as $\bigvee\mathcal{S}$ and $\bigwedge\mathcal{S}$, respectively. 
Specifically,  $\boldsymbol{g}=\bigwedge\mathcal{S}\subset\mathcal{P}^d$ is defined as \cite{FC02S,BG19O}
\begin{align}
\label{eq:meet}
g_i=\min\left\lbrace\sum_{j=1}^{i}u_{j} | \forall\ \boldsymbol{u}\in\mathcal{S}\right\rbrace-
\min\left\lbrace\sum_{j=1}^{i-1}v_{j} | \forall\ \boldsymbol{v}\in\mathcal{S}\right\rbrace\ ,
\end{align}
where $i\in[1,d]$ and $\sum_{j=1}^k u_j=\sum_{j=1}^k v_j=0$ if $k=0$. 

Let $\ket{\psi}=\sum_{i=1}^{d}(\psi)_{i}\ket{i}$ and $\ket{\phi}=\sum_{i=1}^{d}(\phi)_{i}\ket{i}$ be two pure coherent states of which $\abs{(\psi)_1}\geqslant\abs{(\psi)_2}\geqslant\cdots\geqslant\abs{(\psi)_d}$ and $\abs{(\phi)_1}\geqslant\abs{(\phi)_2}\geqslant\cdots\geqslant\abs{(\phi)_d}$, respectively, the coherent vector $\psi$ writes 
\begin{align}
\label{eq:cohVec}
\psi=(\abs{(\psi)_1}^2,\abs{(\psi)_2}^2,\cdots,\abs{(\psi)_d}^2)\ .
\end{align}
We denote the $i$-th component of coherent vector $\psi$ as $\psi_{i}=\abs{(\psi)_i}^2$ hereafter. Accordingly, the coherence transformation between state $\ket{\psi}$ and $\ket{\phi}$ can be deterministically realized under incoherent operations if and only if the following majorization relation holds~\cite{DS15C,ZH17O} 
\begin{equation}
\label{prop:cohe_tran}
\psi\prec\phi\ .
\end{equation}
\noindent Note, the coherence transformation is not achievable when the \cref{prop:cohe_tran} is not satisfied. Very recently, a framework of approximation CT for an initial state $\rho$ and the target state $\ket{\phi}$ was proposed~\cite{LC22A}. For the sake of clarity, we present the pure state to pure state case in the following for illustration \cite{DS15Co, ZH17O,LC22A}.

From the definition of the majorization relation, one can define a set of coherence monotones as follows \cite{DS15Co,ZH17O,LC22A}
\begin{equation}
\label{eq:cohemonot}
C_{l}(\psi)=\sum_{i=l}^{d}\psi_{i}\ , \ \forall\ l\in [1,d]\ .
\end{equation}
The \cref{prop:cohe_tran} can be reexpressed as 
\begin{equation}
\ket{\psi}\stackrel{IO}{\longrightarrow}\ket{\phi}\Longleftrightarrow C_{l}(\psi) \geqslant C_{l}(\phi),\quad \forall\ l\in[1,d]\ .
\end{equation}
Through incoherent operation, the state $\ket{\psi}$ can be deterministically transformed into an intermediate state $\ket{\varphi}$, which maximizes the fidelity with the target state $\ket{\phi}$, i.e.,
\begin{equation}
\label{eq:intfid}
\ket{\varphi^{\mathrm{f}}}=\mathrm{arg}\max\limits_{\ket{\varphi}:\ket{\psi}\rightarrow\ket{\varphi}} \abs{\bra{\varphi}\phi\rangle}^{2}.
\end{equation}

Here, the intermediate state $\ket{\varphi^\mathrm{f}}$ was proved to be identical with the intermediate state of an optimal protocol for PCT, and can be constructed via the following procedure \cite{ZH17O,LC22A,DS15Co}.
First, define $l_{0}=d+1$ and denote
\begin{equation}
\label{eq:optprobtrans}
q_{1}=\frac{C_{l_{1}}(\psi)}{C_{l_{1}}(\phi)}\coloneqq\min_{l\in[1,d]}\frac{C_{l}(\psi)}{C_{l}(\phi)}\ ,
\end{equation}
where $l_1$ symbolizes the smallest integer when $q_1$ is defined.
Note that there exists the case that $q_1=1$ and $l_1=1$ simultaneously is excluded.
Then $q_{j},\ j=2,\cdots,d,$ can be defined by
\begin{equation}
\label{eq:definition-qj}
q_{j}=\frac{C_{l_{j}}(\psi)-C_{l_{j-1}}(\psi)}{C_{l_{j}}(\phi)-C_{l_{j-1}}(\phi)}=\min_{l\in[1,l_{j-1}-1]}\frac{C_{l}(\psi)-C_{l_{j-1}}(\psi)}{C_{l}(\phi)-C_{l_{j-1}}(\phi)}\ .
\end{equation}
Repeating this procedure until the integer $k$ is found so that $l_{k}=1$. It can be verified that $l_{0}>l_{1}>l_{2}>\cdots >l_{k}$ and $0<q_{1}<q_{2}<\cdots<q_{k}$. Hence, one obtains the intermediate state satisfying \cref{eq:definition-qj} by the sequence of ratios $\lbrace q_{j}\rbrace$, with $l_{j}\in\lbrace 1,\cdots,l_{j-1}-1\rbrace$ as the smallest integer when $q_{j}$ is defined, respectively \cite{LC22A}, that is 
\begin{equation}
\label{eq:phifqj}
\ket{\varphi^{\mathrm{f}}}=\sum_{i=1}^{d}(\varphi^{\mathrm{f}})_{i}\ket{i},\ \text{with}\  (\varphi^{\mathrm{f}})_{i}=\sqrt{q_{j}}\,(\phi)_{i}\ ,\ i\in [l_{j},l_{j-1}-1]\ ,\forall\ j\in[1,k]\ .
\end{equation}

From the construction of intermediate state, 
it is easy to find that $\abs{\varphi^\mathrm{f}_i}\geqslant\abs{\varphi^\mathrm{f}_{i+1}}$ and $C_l(\psi)>C_l(\varphi^\mathrm{f})$, i.e., 
\begin{equation}
\psi\prec\varphi^{\mathrm{f}}\ .
\end{equation}
Say, one can transform $\psi$ into $\varphi^{\mathrm{f}}$ through deterministic incoherent operations according to \cref{prop:cohe_tran}. Moreover, it has also been proved that the maximal probability of obtaining the state $\ket{\phi}$ from $\ket{\psi}$, $P_{\max}(\ket{\psi}\rightarrow\ket{\phi})$, is exactly the $q_1$ \cite{DS15Co,ZH17O}.
\begin{align}
\label{prop:maxProb}
P_{\max}(\ket{\psi}\rightarrow\ket{\phi})=\min_{l\in[1,d]}\frac{C_{l}(\psi)}{C_{l}(\phi)}=q_1\ .
\end{align}
Notice that $P_{\max}(\ket{\psi}\rightarrow\ket{\phi})$ is vanishing with the number of nonzero components of $\psi$ being smaller than that of $\phi$. Therefore, we will assume $\psi$ has at least as many nonvanishing components as $\phi$ hereafter.
From \cref{prop:maxProb}, intuitively, the increase of the coherence of target state maybe achieved by the loss of transformation probability, which can be indeed realized according to the following lemma.

\begin{lemma}
\label{thm:maxproschur}
Let $\mathfrak{q}_\mathrm{I}(\phi)$ be the maximal coherence transformation probability function of $\ket{\phi}$ of transforming a prefixed initial state $\ket{\Psi}$ to the target state $\ket{\phi}$, $\mathfrak{q}_\mathrm{T}(\psi)$ be the maximal CT probability function of $\ket{\psi}$ of transforming the initial state $\ket{\psi}$ to a prefixed target state $\ket{\Phi}$ via SIO, i.e.,
\begin{align}
\mathfrak{q}_\mathrm{I}(\phi)\coloneqq\min_{l\in[1,d]}\frac{C_{l}(\Psi)}{C_{l}(\phi)}\ ,\quad \mathfrak{q}_\mathrm{T}(\psi)\coloneqq\min_{l\in[1,d]}\frac{C_{l}(\psi)}{C_{l}(\Phi)}\ ,
\end{align}
 if $\phi\prec\phi'$, the function $\mathfrak{q}_\mathrm{I}(\phi)$ is a Schur-convex function, $\mathfrak{q}_\mathrm{I}(\phi) \leqslant \mathfrak{q}_\mathrm{I}(\phi')$;  if $\psi\prec\psi'$, the function $\mathfrak{q}_\mathrm{T}(\phi)$ is a Schur-concave function, $\mathfrak{q}_\mathrm{T}(\psi) \geqslant \mathfrak{q}_\mathrm{T}(\psi')$\ .
\end{lemma}
\begin{proof}
For two possible target states $\ket{\phi}$ and $\ket{\phi'}$ satisfying $\phi\prec\phi'$, according to the definition of majorization relation and \cref{eq:cohemonot}, we have
\begin{align}
C_l(\phi)\geqslant C_l(\phi')\ ,\ \forall\ l\in[1,d]\ .
\end{align}
Thus, for a prefixed initial state $\ket{\Psi}$, 
\begin{align}
\frac{C_{l}(\Psi)}{C_{l}(\phi)}\leqslant\frac{C_{l}(\Psi)}{C_{l}(\phi')}\ ,\ \forall\ l\in[1,d]\ .
\end{align}
Let $l_1$ and $l_1'$ be the smallest integers for $\mathfrak{q}_\mathrm{I}(\phi)$ and $\mathfrak{q}_\mathrm{I}(\phi')$, respectively, then
\begin{align}
\frac{C_{l_1}(\Psi)}{C_{l_1}(\phi)}\leqslant\frac{C_{l_1'}(\Psi)}{C_{l_1'}(\phi)}\leqslant\frac{C_{l_1'}(\Psi)}{C_{l_1'}(\phi')}\ .
\end{align}
Hence, 
\begin{align}
\mathfrak{q}_\mathrm{I}(\phi) \leqslant \mathfrak{q}_\mathrm{I}(\phi')\ .
\end{align}
Similarly, one can prove $\mathfrak{q}_\mathrm{T}(\psi)$ is a Schur-concave function.
\end{proof}

\section{Probabilistic Coherence Transformation Protocols}
\label{sec:protocols}
From preceding section, we know that the coherence transformation cannot always be realized deterministically. However, the transformation might be accomplished probabilistically. In this section, a scheme of PCT is proposed leveraging the majorization lattice \cite{FC02S}, of which there are two key elements, the meet and join operations. 

Given an initial state $\ket{\psi}$ and a target state $\ket{\phi}$, the optimal common resource (OCR) state, $\ket{\varphi^\wedge}$, can be defined by the meet of coherence vectors as \cite{BG19O}
\begin{align}
\varphi^\wedge=\psi\wedge\phi\ ,
\end{align}
and the optimal common product (OCP) state, $\ket{\varphi^\vee}$, satisfies
\begin{align}
\varphi^\vee=\psi\vee\phi\ .
\end{align}
\noindent Note that the OCR and OCP states can be generalized to the multiple states case. The greedy and thrifty protocols for PCT will be constructed based on OCP and OCR states, respectively. In case of two coherence vectors $\psi$ and $\phi$ are incomparable, from \cref{prop:cohe_tran} one cannot transform them deterministically. However, it can be realized by the greedy and thrifty protocols goes as follows.

\subsection{Probabilistic Coherence Transformation Based on OCP State}

In greedy protocol, one can deterministically transform $\ket{\psi}$ into the OCP state $\ket{\varphi^\vee}$ via SIO due to $\psi\prec\varphi^\vee$. Then, the transformation from $\ket{\varphi^\vee}$ to $\ket{\phi}$ can be readily realized probabilistically through an intermediate state $\ket{\varphi^\mathrm{g}}$, as $\varphi^\vee\nprec\phi$. 

From \cref{eq:intfid}, $\ket{\varphi^\mathrm{f}}$ is defined as the intermediate state that maximizes the fidelity with $\ket{\phi}$. Similarly, it can be verified that the OCP state $\ket{\varphi^\vee}$ is the closest state to $\ket{\phi}$ obtained from $\ket{\psi}$, i.e.,
\begin{align}
\ket{\varphi^\vee}=\mathrm{arg}\min_{\ket{\varphi}:\ket{\psi}\rightarrow\ket{\varphi}}\mathrm{d}(\phi,\varphi)\ .
\end{align}
Here, $\mathrm{d}(\psi,\phi)\coloneqq H(\psi)+H(\phi)-2H(\psi\vee\phi)$ with $H(\psi)$ the Shannon entropy of $\psi$ is a distance measure in majorization lattice theory \cite{FC13I}.
Thus, it is well motivated to compare the coherent states $\ket{\varphi^\vee}$ and $\ket{\varphi^{\mathrm{f}}}$.
Of the greedy protocol, we find the majorization relation between $\varphi^\vee$ and $\varphi^f$ through the following lemma.
\begin{lemma}
\label{lem:aprecra}
For any descending ordered vectors $\boldsymbol{a}\in\mathcal{P}^d$ and $\boldsymbol{r}\in\mathbb{R}^d$ satisfying $\sum_{i=1}^d a_i=\sum_{i=1}^d r_i a_i=1$, we have 
\begin{align}
\label{eq:arhadamarda}
\boldsymbol{a}\prec\boldsymbol{r}\ \odot\ \boldsymbol{a}\ .
\end{align}
Here, $\odot$ denotes the Hadamard product.
\end{lemma}
\begin{proof}
If $\boldsymbol{r}=(1,1,\cdots,1)$, \cref{eq:arhadamarda} obviously holds. For  $\boldsymbol{r}\neq(1,1,\cdots,1)$, from the condition $\sum_{i=1}^d a_i=\sum_{i=1}^d r_i a_i=1$, we have $r_1>1$ and $r_d<1$. Because, 
\begin{align}
a_1+a_2+\cdots+a_d=1\ ,\\
r_1 a_1+r_2 a_2+\cdots+r_d a_d=1\ ,
\end{align}
by subtraction, we readily obtain
\begin{align}
\label{eq:rarsubtrt}
(1-r_1)a_1=(r_2-1)a_2+\cdots+(r_d-1)a_d<0\ .
\end{align}
Clearly, $a_1<r_1 a_1$. Furthermore, provided $r_i\leqslant1$ with $i=2,\cdots,d$, by moving the terms of the right hand side of \cref{eq:rarsubtrt} to the left recursively, one can conclude that 
\begin{align}
\label{eq:rarprec}
L_j(\boldsymbol{a})\leqslant L_j(\boldsymbol{r}\odot\boldsymbol{a}),\quad j=1,\cdots,d\ .
\end{align}
Here, $L_j(\boldsymbol{a})=\sum_{i=1}^j a_i$. Following the similar line of reasoning, \cref{eq:rarprec} can also be proved for $r_i\geqslant1$ with $i=1,2$; $r_i\geqslant1$ with $i=1,2,3$, and so on till $r_i\geqslant1$ with $i=1,2,\cdots,d-1$. Conclusively, we then have $\boldsymbol{a}\prec\boldsymbol{r}\ \odot\ \boldsymbol{a}\ $. 
\end{proof}

\begin{theorem}
\label{thm:optMsup}
Given initial state $\ket{\psi}$ and target state $\ket{\phi}$ of the CT under IO or SIO, the following majorization relation holds
\begin{align}
\label{eq:optMsup}
\varphi^{\vee}\prec\varphi^{\mathrm{f}}\ .
\end{align}
Here, $\varphi^{\vee}$ is OCP state and $\varphi^{\mathrm{f}}$ is defined in \cref{eq:intfid}. 
\end{theorem}
\begin{proof}
According to the construction of $\ket{\varphi^\mathrm{f}}$, $\psi\prec\varphi^\mathrm{f}$. By the definition of the majorization lattice, $\psi^\vee$ is the supremum of $\psi$ and $\phi$, thus, if one can show $\phi\prec\varphi^\mathrm{f}$, then $\varphi^\vee\prec\varphi^\mathrm{f}$. From \cref{eq:phifqj}, we have $ \varphi^{\mathrm{f}}_{i}=q_{j}\phi_{i}\ ,\ i\in [l_{j},l_{j-1}-1]\ ,\forall\ j\in[1,k]$, that is $ \varphi^{\mathrm{f}}=\boldsymbol{m}\odot\phi$ with $\boldsymbol{m}_i=q_j$, $i\in [l_{j},l_{j-1}-1]$. By \cref{lem:aprecra}, we get $\phi\prec\varphi^{\mathrm{f}}$.
\end{proof}

The \cref{thm:optMsup} tells if a CT can be realized  by means of intermediate state $\ket{\varphi^\mathrm{f}}$ under SIO, it can also be obtained via $\ket{\varphi^\vee}$, which implies that the greedy protocol is an optimal one for PCT \cite{LC22A}. Moreover, in the second step of the greedy protocol, i.e., the PCT $\ket{\varphi^\vee}\stackrel{SIO}{\longrightarrow}\ket{\phi}$, the intermediate state $\ket{\varphi^\mathrm{g}}$ is identical to $\ket{\varphi^\mathrm{f}}$, which can be deduced from $\psi\prec\varphi^{\vee}\prec\varphi^{\mathrm{f}}$. Thus, by merging the two deterministic steps $\ket{\psi}\stackrel{IO}{\longrightarrow}\ket{\varphi^\vee}\stackrel{SIO}{\longrightarrow}\ket{\varphi^\mathrm{g}}$ into a single one $\ket{\psi}\stackrel{SIO}{\longrightarrow}\ket{\varphi^\mathrm{g}}$, the greedy protocol reduces to protocol given in Ref. \cite{DS15Co,LC22A} . 

\subsection{Probabilistic Coherence Transformation Based on OCR State}
In thrifty protocol, the coherence transformation precedure proceeds oppsitely. Explicitly, one first transforms $\ket{\psi}$ to the OCR state, $\ket{\varphi^\wedge}$, probabilistically through an intermediate state $\ket{\varphi^\mathrm{t}}$, as $\psi\nprec\varphi^\wedge$. Then, transforms $\ket{\varphi^\wedge}$ to $\ket{\phi}$ deterministically, which is feasible because $\varphi^\wedge\prec\phi$. Notably, the thrifty protocol better preserves average coherence than greedy protocol, which will be proved in the following. 
To proceed, the coherence monotone of OCR state can be determined by following lemma.
\begin{lemma}
\label{lema:infimum-cohe}
The coherence monotone of the OCR state $\ket{\varphi^\wedge}$ for $m$ pure states $\{\ket{\psi^\iota}\}_{\iota=1}^m$ can be determined by
\begin{align}
C_l(\varphi^\wedge)=\max_{\iota\in[1,m]}\{C_l(\psi^\iota)\},\quad \forall\ l\in[1,d]\ .
\end{align}
Here, $C_l(\psi^\iota)$ represents the coherence monotone of the state $\ket{\psi^\iota}$.
\end{lemma}
\begin{proof}
The $l$-th coherence monotone of the OCR state, $\ket{\varphi^\wedge}$, writes $C_l(\varphi^\wedge)=1-\sum_{i=1}^{l-1}\varphi^\wedge_i$. By \cref{eq:meet}, 
\begin{align}
C_l(\varphi^\wedge)
&=1-\sum_{i=1}^{l-1}\left(\min_{\iota\in[1,m]}\left\{\sum_{j=1}^{i}\psi^\iota_j\right\}-\min_{\iota\in[1,m]}\left\{\sum_{j=1}^{i-1}\psi^\iota_j\right\}\right)\\
&=1-\min_{\iota\in[1,m]}\left\{\sum_{i=1}^{l-1}\psi^\iota_j\right\}\ .
\end{align}
Thus, we have
\begin{align}
C_l(\varphi^\wedge)=\max_{\iota\in[1,m]}\left\{\sum_{j=l}^{d}\psi^\iota_j\right\}=\max_{\iota\in[1,m]}\{C_l(\psi^\iota)\}\ .
\end{align}
\end{proof}

Intuitively, the OCR state also has a natural interpretation, in which, for simplicity, we use a distance measure $\mathrm{D}(\psi,\phi)$ \cite{FC13I} that is dual to the $\mathrm{d}(\psi,\phi)$.
\begin{lemma}
\label{thm:closeststate}
Given initial state $\ket{\psi}$ and target state $\ket{\phi}$ of the CT under SIO, the OCR state $\ket{\varphi^\wedge}$ is the closest state to $\ket{\psi}$ among the states that can be transformed to $\ket{\phi}$, that is 
\begin{align}
\ket{\varphi^\wedge}=\mathrm{arg}\min_{\ket{\varphi}:\ket{\varphi}\rightarrow\ket{\phi}}\mathrm{D}(\psi,\varphi)\ .	
\end{align}
Here, $\mathrm{D}(\psi,\varphi)\coloneqq G(\psi)+G(\varphi)-2G(\psi\wedge\varphi)$ with the Gini Index $G(\psi)=\frac{d+1}{d}-\frac{2}{d}\sum_{i=1}^d i\psi_i$, is a distance measure in majorization lattice theory.
\end{lemma}
\begin{proof}
We notice that
\begin{align}
\label{eq:distancewedge}
\mathrm{D}(\psi,\phi)=\mathrm{D}(\psi,\phi\wedge\psi)+\mathrm{D}(\phi,\phi\wedge\psi)\ ,
\end{align}
it is sufficient to prove the lemma by following three cases:
(1) If $\varphi\prec\varphi^\wedge\prec\phi$,  from \cref{eq:distancewedge}, it is readily to obtain $\mathrm{D}(\psi,\varphi)=\mathrm{D}(\psi,\varphi^\wedge)+\mathrm{D}(\varphi^\wedge,\phi)\geqslant\mathrm{D}(\psi,\varphi^\wedge)$.
(2) For $\varphi^\wedge\prec\varphi\prec\phi$, one can conclude that $\psi\wedge\varphi=\varphi^\wedge$, which also leads to  $\mathrm{D}(\psi,\varphi)=\mathrm{D}(\psi,\varphi^\wedge)+\mathrm{D}(\varphi^\wedge,\phi)\geqslant\mathrm{D}(\psi,\varphi^\wedge)$.
(3) When $\varphi\prec\phi$, $\varphi\nprec\varphi^\wedge$ and $\varphi^\wedge\nprec\varphi$, the relation $\psi\wedge\varphi\prec\varphi^\wedge\prec\psi$ holds, we then have $\mathrm{D}(\psi,\varphi)=\mathrm{D}(\psi,\psi\wedge\varphi)+\mathrm{D}(\psi\wedge\varphi,\varphi)\geqslant \mathrm{D}(\psi,\psi\wedge\varphi)\geqslant \mathrm{D}(\psi,\varphi^\wedge)$.
\end{proof}
\noindent Considering \cref{thm:maxproschur}, the \cref{thm:closeststate} implies that by transforming $\ket{\psi}$ to the closest one among the states that can be transformed to $\ket{\phi}$, one may achieve the maiximal probability of CT through the least coherence gain, as shown in the following lemma.

To construct the intermediate state $\ket{\varphi^\mathrm{t}}$ between $\ket{\psi}$ and $\ket{\varphi^\wedge}$, a sequence of ratios $\{t_j\}$ maybe defined, in which $t_1=\mathfrak{q}_\mathrm{I}(\varphi^\wedge)$ is the maximal CT probability.

\begin{lemma}
\label{thm:probtransmeet}
Let $\ket{\psi}$ and $\ket{\phi}$ be the initial state and target state of CT via SIO, respectively, the maximal probability of CT from $\ket{\psi}$ to $\ket{\varphi^\wedge}$ reaches the maximum of CT from $\ket{\psi}$ and $\ket{\phi}$, that is $t_1=q_1$. Moreover, $t_j=q_j$ holds for all $j\in[1,k]$.
\end{lemma}
\begin{proof}
From \cref{eq:optprobtrans}, we have
\begin{align}
t_1=\min_{l\in[1,d]}\frac{C_l(\psi)}{C_l(\varphi^\wedge)}\equiv\frac{C_{l^{\prime}_1}(\psi)}{C_{l^{\prime}_1}(\varphi^\wedge)}\ ,
\end{align}
with $l^{\prime}_1$ being the smallest $l\in[1,d]$ minimizing $\frac{C_l(\psi)}{C_l(\varphi^\wedge)}$\ . 
Since by \cref{lema:infimum-cohe}
\begin{align}
t_1=\min_{l\in[1,d]}\frac{C_l(\psi)}{\max\{C_l(\psi),C_l(\phi)\}}\ ,
\end{align}
for any $l$ satisfying $C_l(\phi)<C_l(\psi)$, we have
\begin{align}
\frac{C_l(\psi)}{C_l(\phi)}>\frac{C_l(\psi)}{\max\{C_l(\psi),C_l(\phi)\}}=1\ .
\end{align}
Therefore, in this case the $q_1$ is not well defined, and hence will not being considered. That is to say the minimization procedure only happens for those $l$ guaranteeing $C_l(\phi)\geqslant C_l(\psi)$, i.e.,
\begin{align}
\label{eq:maxminiwdg}
\frac{C_l(\psi)}{\max\{C_l(\psi),C_l(\phi)\}}=\frac{C_l(\psi)}{C_l(\phi)}\ .
\end{align}
Then, we get
\begin{align}
\label{eq:t1q1l1}
t_1=q_1\ ,\quad l^{\prime}_1=l_1\ .
\end{align}

From \cref{eq:definition-qj,eq:maxminiwdg,eq:t1q1l1},  we have
\begin{align}
t_2=\min_{l\in[1,l_1-1]}\frac{C_l(\psi)-C_{l_1}(\psi)}{\max\{C_{l}(\psi),C_{l}(\phi)\}-C_{l_1}(\phi)}\ .
\end{align}
It follows from the equivalence
\begin{align}
a\geqslant b\ \Longleftrightarrow\ \frac ab\geqslant\frac{a+c}{b+c}\ ,\quad(a,b,c>0)\ ,
\end{align}
that for any $l\in[1,l_1-1]$ giving $C_{l}(\psi)>C_{l}(\phi)$, the following inequality holds
\begin{align}
\frac{C_l(\psi)-C_{l_1}(\psi)}{C_l(\psi)-C_{l_1}(\phi)}>\frac{C_1(\psi)-C_{l_1}(\psi)}{C_1(\psi)-C_{l_1}(\phi)}\ .
\end{align}
Hence, $t_2$ could be obtained as $C_{l}(\psi)\leqslant C_{l}(\phi)$, i.e.,
$t_2=q_2\ ,\ l^{\prime}_2=l_2$.
Similarly, one can prove the remaining equalities for each $j$, i.e.,
\begin{align}
t_j=q_j\ ,\quad l^{\prime}_j=l_j\ ,\quad \forall\ j\in[1,k]\ .
\end{align}
\end{proof}
The \cref{thm:probtransmeet} indicates that the thrifty protocol is also optimal for PCT, by which we can prove that the thrifty protocol preserves more coherence on average than the greedy protocol and the one in Refs. \cite{DS15Co,LC22A}.
\begin{theorem}
\label{thm:residualcoherence}
Given the initial state $\ket{\psi}$ and target state $\ket{\phi}$, let $\ket{\mu}$ and $\ket{\nu}$ represent those states different from  $\ket{\phi}$ and $\ket{\varphi^\wedge}$ after CT via SIO from $\ket{\varphi^\vee}$ and $\ket{\psi}$, respectively. The state $\ket{\nu}$ is more coherent than $\ket{\mu}$, that is
\begin{align}
\nu\prec\mu\ .
\end{align}
\end{theorem}
\begin{proof}
From \cref{eq:phifqj}, the coherence vector of the intermediate state writes 
\begin{align}
\varphi^{\mathrm{f}}_{i}=q_{j}\phi_{i}\ ,\ i\in [l_{j},l_{j-1}-1]\ ,\forall\ j\in[1,k]\ \Longleftrightarrow\ \varphi^\mathrm{f}=\boldsymbol{m}\,\odot\,\phi\ \ .
\end{align}	
Note, $\ket{\varphi^\mathrm{g}}=\ket{\varphi^\mathrm{f}}$ between $\ket{\varphi^\vee}$ and $\ket{\phi}$. Here, $\boldsymbol{m}_i=q_j$, if $i\in[l_j,l_{j-1}-1]$, and $\odot$ denotes the Hadamard product.
From \cref{thm:probtransmeet}, we know that $t_j=q_j\ , l^{\prime}_j=l_j\ ,\forall\ j\in[1,k]$, and similarly the coherence vector of the intermediate state $\ket{\varphi^\mathrm{t}}$ between the $\ket{\psi}$ and $\ket{\varphi^\wedge}$ can be expressed as 
\begin{align}
\varphi^\mathrm{t}=\boldsymbol{m}\,\odot\,\varphi^\wedge\ .
\end{align}
By means of the fact that for any $\boldsymbol{a}\prec\boldsymbol{b}\in\mathcal{P}^d$ and descending ordered vectors $\boldsymbol{r}\in\mathbb{R}^d$ satisfying $\sum_{i=1}^{d}r_i a_i=\sum_{i=1}^{d}r_i b_i=1$, the relation $\boldsymbol{r}\,\odot\,\boldsymbol{a}\prec\boldsymbol{r}\,\odot\,\boldsymbol{b}$ holds \cite{MA11I}. Thus, $\varphi^\wedge\prec\phi$ leads to
\begin{align}
\label{eq:tImPrecgIm}
\varphi^\mathrm{t}\prec\varphi^\mathrm{f}\ .
\end{align}
\cref{eq:tImPrecgIm} shows that the intermediate state $\ket{\varphi^\mathrm{t}}$ of the thrifty protocol is more coherent than the intermediate state $\ket{\varphi^\mathrm{f}}$. 

Notice that the CT with maximal probability from intermediate state to target state can be performed by the following SIO \cite{VG99E,DS15Co}
\begin{align}
\label{eq:sio}
M=
\begin{pmatrix}
M_k&&\\
&\ddots&\\
&&M_1
\end{pmatrix}.
\end{align}
Here, $M_j=\sqrt{\frac{q_1}{q_j}}\mathbb{I}_{[l_{j-1}-l_j]}\ ,\ \forall\ j\in[1,k]$ with $\mathbb{I}_{[l_{j-1}-l_j]}$ the identity operator in $(l_{j-1}-l_j)$-dimensional Hilbert space. $M$ defines the operation satisfying that
$M\ket{\varphi^\mathrm{f}}=\sqrt{q_1}\ket{\phi}$ and $M\ket{\varphi^\mathrm{t}}=\sqrt{q_1}\ket{\varphi^\wedge}$\ .
When coherence transformation fails, the states $\ket{\mu}$ and $\ket{\nu}$ are yielded from the operation $N$ defined by $N^\dagger N=\mathbb{I}_{[d]}-M^\dagger M$, i.e.,
\begin{align}
\sqrt{1-q_1}\ket{\mu}=N\ket{\varphi^\mathrm{f}},\quad \sqrt{1-q_1}\ket{\nu}=N\ket{\varphi^\mathrm{t}}.
\end{align}
Thus, the coherence vectors of  residual states can be repspectively written as 
\begin{align}
\nu=\boldsymbol{n}\,\odot\,\varphi^\mathrm{t}\ ,\quad \mu=\boldsymbol{n}\,\odot\,\varphi^\mathrm{f}\ ,
\end{align}
where $\boldsymbol{n}_i=\frac{\bra{i}(\mathbb{I}_{[d]}-M^\dagger M)\ket{i}}{(1-q_1)},\ \forall\ i\in[1,d]$\ .
Similarly, one can obtain 
\begin{align}
\nu\prec\mu\ .
\end{align}
\end{proof}
\paragraph{Example} 
To illustrate \cref{thm:residualcoherence}, let us consider a qutrit system where the initial state $\ket{\psi}$ and the target state $\ket{\phi}$ are given as:
\begin{align}
\ket{\psi}=\sqrt{0.5}\ket{1}+\sqrt{0.4}\ket{2}+\sqrt{0.1}\ket{3}\ ,\quad \ket{\phi}=\sqrt{0.7}\ket{1}+\sqrt{0.15}\ket{2}+\sqrt{0.15}\ket{3}\ .
\end{align}
According to \cref{eq:definition-qj,prop:maxProb}, one can calculate that $q_1=2/3$, $q_2=18/17$, and the optimal probability of the PCT from $\ket{\psi}$ to $\ket{\phi}$ is $q_1=2/3$. Thus, by \cref{eq:sio}, the coherence transformations with maximal probability from intermediate state to target state can be performed by the following SIO
\begin{align}
M=
\begin{pmatrix}
\sqrt{\frac{q_1}{q_2}}&0&0\\
0&\sqrt{\frac{q_1}{q_2}}&0\\
0&0&1
\end{pmatrix}.
\end{align}
When coherence transformations fail, that is, one performs the operation $N$ defined by $N^\dagger N=\mathbb{I}_{[3]}-M^\dagger M$, the unnormalized residual states of greedy and thrifty protocols are, respectively,
\begin{align}
\ket{\mu}=\sqrt{0.7}\ket{1}+\sqrt{0.15}\ket{2}\ ,\quad \ket{\nu}=\sqrt{0.5}\ket{1}+\sqrt{0.35}\ket{2}\ .
\end{align}
It can be easily verified that $\nu\prec\mu$, which illustrates \cref{thm:residualcoherence}.
\section{Probabilistic Coherence Transformation of Mixed States}
\label{sec:protranmix}
In preceding section, the probabilistic coherence transformation between pure states is investigated. Now we extend the procedure to the study of pure state ensembles or mixed states.

\subsection{Probabilistic Coherence Transformation to Pure State Ensemble}
Provided there are $m$ possible target states $\{\ket{\phi^\iota}\}_{\iota=1}^m$, each of them may appear with certain possibility after the CT from a pure initial state. Define the OCR state $\ket{\varphi^\wedge}$ configured by the $m+1$ coherence vectors as 
\begin{align}
\varphi_{\mathrm{e}}^{\wedge}\coloneqq\psi\wedge\phi^1\wedge\cdots\wedge\phi^m\ ,
\end{align}
then one can prove that the maximal CT probability from $\ket{\psi}$ to $\ket{\varphi_{\mathrm{e}}^{\wedge}}$ is upper bounded by the maximal probabilities of the $m$ possible individual transformations from $\ket{\psi}$ to $\ket{\phi^\iota}$.
\begin{corollary}
\label{thm:thrmulti}
Given an initial state $\ket{\psi}$ and a target state ensemble $\{\ket{\phi^\iota}\}_{\iota=1}^m$\ , the maximal probability of coherence transformation from $\ket{\psi}$ to the OCR state $\varphi_{\mathrm{e}}^{\wedge}$ via SIO cannot exceed the maximal probabilities of the $m$ possible individual CTs from $\ket{\psi}$ to $\ket{\phi^\iota}$, that is
\begin{align}
\mathfrak{q}_{\mathrm{I}}(\varphi_{\mathrm{e}}^{\wedge})\leqslant\mathfrak{q}_{\mathrm{I}}(\phi^\iota)\ ,\quad\forall\ \iota\in[1,m]\ .
\end{align}
The inequality will be saturated by the minimal $\mathfrak{q}_{\mathrm{I}}(\phi^\iota)$ among $m$ transformations from $\ket{\psi}$ to $\ket{\phi^\iota}$.
\end{corollary}
\begin{proof}
According to the property of majorzation lattice, $\varphi_{\mathrm{e}}^{\wedge}\prec\varphi^{\wedge,\iota}$, $\forall\ \iota\in[1,m]$ with $\varphi^{\wedge,\iota}$ the coherence vector of the OCR state between $\ket{\psi}$ and $\ket{\phi^\iota}$, from \cref{thm:maxproschur,thm:probtransmeet}, we get
\begin{align}
\mathfrak{q}_{\mathrm{I}}(\varphi_{\mathrm{e}}^{\wedge})\leqslant\mathfrak{q}_{\mathrm{I}}(\varphi^{\wedge,\iota})=\mathfrak{q}_{\mathrm{I}}(\phi^\iota)\ ,\quad\forall\ \iota\in[1,m]\ .
\end{align}
According to \cref{lema:infimum-cohe}, the coherence monotone of $\ket{\varphi^\wedge_\mathrm{e}}$ can be determined as
\begin{align}
\label{eq:chrmntnesb}
\mathfrak{q}_{\mathrm{I}}(\varphi_{\mathrm{e}}^{\wedge})=\min_{l\in[1,d]}\frac{C_l(\psi)}{\max\{C_l(\psi),C_l(\phi^1),\cdots,C_l(\phi^\iota),\cdots,C_l(\phi^m)\}}\ .
\end{align}
The right hand side of \cref{eq:chrmntnesb} reaches minimum as $C_l(\psi)>C_l(\phi^\iota)$, $\forall\ \iota$, it leads to a deterministic CT. Thus, for probabilistic coherence transformation, it is sufficient to consider
\begin{align}
\mathfrak{q}_{\mathrm{I}}(\varphi_{\mathrm{e}}^{\wedge})=\min_{l\in[1,d],\iota\in[1,m]}\ \frac{C_l(\psi)}{C_l(\phi^\iota)} \ .
\end{align}
Conclusively, we have
\begin{align}
\mathfrak{q}_{\mathrm{I}}(\varphi_{\mathrm{e}}^{\wedge})=\min_{\iota\in[1,m]}\mathfrak{q}_{\mathrm{I}}(\phi^\iota)\ .
\end{align}
\end{proof}
\noindent The \cref{thm:thrmulti} shows that the maximal transformation probability from $\ket{\psi}$ to $\ket{\varphi_{\mathrm{e}}^\wedge}$ is less than the least maximal transformation probability among $m$ possible transformations from $\ket{\psi}$ to $\ket{\phi^\iota}$. Note that an ensemble $\{p_j,\ket{\phi_j}\}_{j=1}^m$ is obtainable from $\ket{\psi}$ via IO iff $\psi\prec\sum_{j=1}^m p_j \phi_j$ \cite{QX15C,ZH17O}, we know $\ket{\phi_j}$ is generated with probability $p_j$. In contrast, if the OCR state is yielded with a fixed probability, any target state $\ket{\phi_j}$ selected at will can be deterministically obtained by thrifty protocol.

\subsection{Coherence Transformation between Mixed States}
In this subsection, we study the deterministic and probabilistic coherence transformations between two arbitrary mixed states, which are complicated and remain unsolved generally. 

Notice that the deterministic CT between mixed states $\rho$ and $\sigma$ maybe realized via SIO through pure intermediate states, we have

\begin{lemma}
\label{lem:sufcon}
Given any two mixed states $\rho$ and $\sigma$, $\rho$ can be deterministically transformed into $\sigma$ if there exists an orthogonal complete set of incoherent projectors $\{P_\alpha\}_{\alpha=1}^r$ satisfying
\begin{align}
\label{eq:sufcon1}
\frac{P_\alpha\rho P_\alpha}{\mathrm{Tr}[P_\alpha\rho P_\alpha]}&=\ket{\psi^\alpha}\bra{\psi^\alpha}\quad \text{and}\quad \psi^\alpha\prec\varphi^\sigma\ ,\ \forall\ \alpha\in[1,r].
\end{align}
Here, $\varphi^\sigma=(\sigma_{11},\cdots,\sigma_{dd})$.
\end{lemma}
\begin{proof}
If the condition \cref{eq:sufcon1} is satisfied, the mixed state $\rho$ can be transformed to $\ket{\varphi^\sigma}$ via SIO \cite{LC19D}, where $\ket{\varphi^\sigma}$ is defined from $\varphi^\sigma=(\sigma_{11},\cdots,\sigma_{dd})$ with $\sigma_{ii}$ the diagonal elements of $\sigma$. Then the coherence transformation from $\ket{\varphi^\sigma}$ to $\sigma$ can be achieved under SIO according to Theorem 3.3 in Ref. \cite{BZ24S}.
\end{proof}

Now we explore the PCT between mixed states by adapting the greedy protocol. As mentioned above, the probabilistic coherence transformation between mixed states maybe converted to the PCT from a set of pure states to a pure target state $\ket{\varphi^\sigma}$, that is 
\begin{align}
\{\ket{\psi^\alpha}\}_{\alpha=1}^r\stackrel{SIO}{\longrightarrow}\ket{\varphi^\sigma}\ ,
\end{align}
of which the OCP state $\ket{\varphi_{\mathrm{m}}^\vee}$ maybe defined as $\varphi_{\mathrm{m}}^\vee=\varphi^\sigma\vee\psi^1\vee\cdots\vee\psi^\alpha\vee\cdots\vee\psi^r$.
Then, from \cref{thm:maxproschur}, given the target state $\ket{\varphi^\sigma}$, we immediately get
\begin{align}
\mathfrak{q}_{\mathrm{T}}(\varphi_{\mathrm{m}}^\vee)\leqslant\mathfrak{q}_{\mathrm{T}}(\varphi^{\vee,\alpha})\ .
\end{align}
Here, $\varphi^{\vee,\alpha}$ denotes the coherence vector of the OCP state between $\ket{\psi^\alpha}$ and $\ket{\varphi^\sigma}$.
Say, the maximal CT probability from $\ket{\varphi_{\mathrm{m}}^\vee}$ to $\ket{\varphi^\sigma}$ is smaller than the least maximal CT probability among the $r$ possible transformations from $\ket{\psi^\alpha}$ to $\ket{\varphi^\sigma}$. 

One would expect that $\mathfrak{q}_{\mathrm{T}}(\varphi_{\mathrm{m}}^\vee)=\min\limits_{\alpha\in[1,r]}\mathfrak{q}_{\mathrm{T}}(\varphi^{\vee,\alpha})$, which, however, may not always hold, since the join of the vectors cannot be simply determined like their meet \cite{FC02S}. Explicitly, $\varphi_{\mathrm{m}}^\vee$ can be written as $\varphi_{\mathrm{m}}^\vee=\bigvee\{\varphi^{\vee,\alpha}\}_{\alpha=1}^r$ according to the property of lattice \cite{DB02I}, for which the coherence monotone $C_l(\varphi_{\mathrm{m}}^\vee)=\min\limits_{\alpha\in[1,r]}\{C_l(\varphi^{\vee,\alpha})\}$ does not hold in general. It is noteworthy that the OCP state $\ket{\varphi_{\mathrm{m}}^\vee}$ acts like a catalyst state. By adding the OCP state $\ket{\varphi_{\mathrm{m}}^\vee}$ to the coherence transformation beforehand, one can first perform the PCT from $\ket{\varphi_{\mathrm{m}}^\vee}$ to $\ket{\phi}$. Subsequently, $\ket{\varphi_{\mathrm{m}}^\vee}$ can be recovered through deterministic CT from $\{\ket{\psi^{\alpha}}\}_{\alpha=1}^r$ to $\ket{\varphi_{\mathrm{m}}^\vee}$, that is $\{\ket{\psi^{\alpha}}\}_{\alpha=1}^r\stackrel{\ket{\varphi_{\mathrm{m}}^\vee}}{\longrightarrow}\ket{\phi^\sigma}$.

\section{Probabilistically Converting Coherence into Entanglement}
\label{sec:coheent}
In this section, as an application, we show that the previous results for PCT can be applied to the probabilistic conversion from coherence into entanglement. 

Coherence and entanglement are both intuitive measures of a system's quantumness, rooted in the superposition principle. 
It was proved that any degree of coherence can be converted into entanglement through incoherent operation \cite{AJ05C,SA15M}. 
Let $\ket{\Phi}=\sum_{j=0}^{d-1}\sqrt{\lambda_j}\ket{j}\otimes\ket{j}$ be a bipartite entangled state with $\lambda_\Phi=(\lambda_1,\cdots,\lambda_d)$ a descending ordered Schimidt vector, then a coherent state $\ket{\psi}$ can be transformed to $\ket{\Phi}$ via IO or SIO iff $\psi\prec\lambda_\Phi$ by the Theorem 5 of \cite{ZH17O}, that is $\ket{\Phi}=\Lambda_\mathrm{i}(\ket{\psi}\otimes\ket{0})$ with $\Lambda_\mathrm{i}$ an incoherent operation. Specifically, when $\psi\prec\lambda_\Phi$, the conversion can be accomplished through the coherence transformation from $\ket{\psi}$ to $\ket{\phi}=\sum_{j=0}^{d-1}\sqrt{\lambda_j}\ket{j}$ via IO followed by a generalized CNOT gate $U_{CNOT}$ on $\ket{\phi}$ with an incoherent ancilla, i.e., $\ket{\Phi}=U_{CNOT}(\ket{\phi}\otimes\ket{0})$. Here, $U_{CNOT}$ is also an incoherent operation \cite{SA15M,ZH17O}. 

The probabilistic conversion from coherence to entanglement can be improved by dint of thrifty protocol of PCT. The optimal common resource state $\ket{\varphi^\wedge_{\mathrm{E}}}$ of both coherence and entanglement maybe defined according to
\begin{align}
\varphi^\wedge_{\mathrm{E}}=\psi\wedge\lambda_\Phi\ .
\end{align}
It was proved that the maximal probability of generating $\ket{\Phi}$ from $\ket{\psi}$ by IO or SIO with the aid of an ancilla writes \cite{ZH17O}
\begin{align}
P_{\max}(\ket{\psi}\rightarrow\ket{\Phi})=\min_{l\in[0,d-1]}\frac{\sum_{l}^{d-1}\psi_l}{\sum_{l}^{d-1}\lambda_l}\ .
\end{align}
Obviously, $P_{\max}(\ket{\psi}\rightarrow\ket{\Phi})=P_{\max}(\ket{\psi}\rightarrow\ket{\phi})$.  According to \cref{thm:probtransmeet}, the maximal CT probability from $\ket{\psi}$ to $\ket{\varphi^\wedge_{\mathrm{E}}}$, a PCT, equals to  $P_{\max}(\ket{\psi}\rightarrow\ket{\phi})$. Thus, the thrifty protocol is optimal for this conversion. Moreover, when the conversion fails, one may get the entangled state
\begin{align}
\ket{\Phi^\nu}=U_{CNOT}(\ket{\nu}\otimes\ket{0})\ .
\end{align}
As indicated by \cref{thm:residualcoherence}, the thrifty protocol preserves more coherence as well as entanglement on average for the conversion from coherence into entanglement.
\section{Conclusions}
\label{sec:con}
To summarize, in this work, two optimal protocols, the greedy and thrifty protocols, for PCT are constructed, in which both the OCP and OCR states have natural interpretations on the majorization lattice. When the initial state $\ket{\psi}$ and target state $\ket{\phi}$ satisfy $\psi\prec\phi$, both the two protocols reproduce deterministic CT. Moreover, the thrifty protocol shows superiority in preserving the coherence resource when the transformation falis. Furthermore, the greedy and thrifty protocols can be generalized to the PCTs of pure state ensembles, through which the PCT between mixed states maybe achieved. Specifically, it is proved that the maximal CT probability function of target state's coherence vector for a prefixed initial state, and the function of initial state's coherence vector for a prefixed target state are Schur-convex and Schur-concave functions, respectively. That is to say, the large coherence gain can be realized by the sacrifice of coherence transformation probability, and vice versa. As an application, it is shown that the probabilistic conversion from coherence to entanglement may also benefit from the results of PCT.
\section*{Acknowledgements}
\noindent
This work was supported in part by National Natural Science Foundation of China(NSFC) under the Grants 12475087 
and 12235008, and University of Chinese Academy of Sciences.


\end{document}